\documentclass[10pt, conference, compsocconf]{IEEEtran}
\usepackage{amsmath}
\usepackage{amssymb}
\usepackage{comment}
\usepackage{amsthm}
\usepackage{url}
\usepackage{blkarray}
\usepackage{cite}
\usepackage{algorithm}
\usepackage{algorithmic}

\newtheorem{proposition}{Proposition}

\newtheorem{lemma}{Lemma}

\theoremstyle{definition}
\newtheorem{definition}{Definition}

\renewcommand{\vec}[1]{\mathbf{#1}} 
\newcommand{\vecg}[1]{\boldsymbol{#1}} 
\newcommand{\mat}[1]{\mathbf{#1}} 
\usepackage{color}

\newcommand{\sect}{Section~}
\usepackage{graphicx}

\begin{document}

\title{Resilience Bounds of Sensing-Based Network Clock Synchronization}

\author{\IEEEauthorblockN{Rui Tan $\qquad$ Linshan Jiang $\qquad$ Arvind Easwaran $\qquad$ Jothi Prasanna Shanmuga Sundaram}
\IEEEauthorblockA{School of Computer Science and Engineering, Nanyang Technological University, Singapore}
}

\maketitle

\begin{abstract}
  Recent studies exploited external periodic synchronous signals to synchronize a pair of network nodes to address a threat of delaying the communications between the nodes. However, the sensing-based synchronization may yield faults due to nonmalicious signal and sensor noises. This paper considers a system of $N$ nodes that will fuse their peer-to-peer synchronization results to correct the faults. Our analysis gives the lower bound of the number of faults that the system can tolerate when $N$ is up to 12. If the number of faults is no greater than the lower bound, the faults can be identified and corrected. We also prove that the system cannot tolerate more than $N-2$ faults. Our results can guide the design of resilient sensing-based clock synchronization systems.
\end{abstract}


\begin{IEEEkeywords}
Clock synchronization, fault tolerance, bounds
\end{IEEEkeywords}

\IEEEpeerreviewmaketitle

\section{Introduction}
\label{sec:intro}

For distributed systems such as sensor networks, accurate clock synchronization among the distributed nodes is important. Correct timestamps make sense data; synchronized clocks enable punctual coordinated operations among the nodes. In contrast, desynchronized clocks will undermine system performance and even lead to physical damages and system disruptions in time-critical systems. However, various factors present significant challenges to maintain resilient clock synchronization of distributed systems, such as large network sizes, deep embedding of the nodes into complex physical environments with various disturbances, and exposure of the systems to cybersecurity threats.



Network Time Protocol (NTP) \cite{mills1991internet} is the foremost means of clock synchronization that is widely known and adopted. Its design principle of estimating the offset between the clocks of a pair of nodes based on the network transmission delays of the synchronization packets is also a basis for many other clock synchronization protocols such as Precision Time Protocol (PTP) \cite{4579760} for industrial Ethernets and RBS \cite{elson2002fine}, TPSN \cite{ganeriwal2003timing}, and FTSP \cite{maroti2004flooding} for sensor networks. However, as discussed in RFC 7384 \cite{rfc7384}, the NTP principle is susceptible to various cybersecurity threats. While most of the vulnerabilities can be solved by conventional security measures such as cryptographic authentication and encryption, a simple packet delay attack that delays the transmissions of the synchronization packets has remained as an open issue that cannot be solved by conventional security measures \cite{rfc7384,mizrahi2012game,ullmann2009delay}.

To address the packet delay attack, our previous studies \cite{viswanathan2016exploiting,rabadi2017taming} have developed sensing-based clock synchronization approaches exploiting external periodic signals that are practically difficult for the attacker to tamper with or jam. Specifically, in \cite{viswanathan2016exploiting}, the minute fluctuations of the power grid voltage cycle lengths, which are similar across a geographic area served by the same power grid, are used as a time fingerprint to develop a clock synchronization approach that is secure against the packet delay attack. In \cite{rabadi2017taming}, the power grid voltage phase, which is nearly identical anytime within a city-scale power grid, is integrated into the NTP principle and achieve the security against the packet delay attack as long as a verifiable condition is satisfied.

These sensing-based approaches focus on the peer-to-peer (p2p) clock synchronization for a node pair. Although they well address the cybersecurity concern regarding the packet delay attack, they may be susceptible to the process noises of the external signals and sensor hardware noises/faults. For instance, as shown in \cite{viswanathan2016exploiting}, an insufficiently long time fingerprint may lead to faults in estimating the clock offset between a pair of nodes. In \cite{rabadi2017taming}, when the round-trip time of an NTP synchronization session exceeds twice of the power grid voltage cycle, the approach will yield multiple clock offset estimates, causing ambiguity.
Given the criticality of trustworthy clock synchronization, it is important to develop methods with understood resilience bounds to deal with the nonmalicious synchronization faults of the sensing-based clock synchronization approaches.

In this paper, based on a general class of p2p sensing-based clock synchronization, we study the resilience of {\em network clock synchronization} for a network of $N$ nodes against the p2p synchronization faults. Upon the occurrence of a fault between a pair of nodes, the measured offset between the two nodes' clocks will have an error of a multiple of the period of the used external signal. In the network clock synchronization, every node pair in the network performs a p2p clock synchronization session and returns the measured clock offset to a central node.
Based on a total of $N \choose 2$ clock offset measurements, the central node uses an algorithm
to estimate the offsets of all nodes' clocks from a selected reference node's clock, while accounting for the possible p2p synchronization faults. Specifically, each step of the algorithm assumes that $k$ out of totally $N \choose 2$ p2p synchronization sessions are faulty, exhaustively tests all possible ${N \choose 2} \choose k$ distributions of these faulty p2p synchronization sessions, and yields a solution once the estimated clock offsets and the estimated p2p clock synchronization faults agree with all the p2p clock offset measurements. Starting from $k=0$, the algorithm increases $k$ by one in each step and terminates once a solution is found. Thus, this algorithm does not require any run-time knowledge about the p2p synchronization faults, including the number of the faults and their distribution among the $N \choose 2$ p2p synchronization sessions.

\begin{table}
  \centering
  \caption{Lower bound of tolerable faults.}
  \label{tab:lower-bound}
  \begin{tabular}{c|ccccccccc}
    \hline
    $N$ & 4 &  5 & 6 & 7 & 8 & 9 & 10 & 11 & 12 \\
    \hline
    Lower bound of & 1 & 1 & 2 & 2 & 2 & 3 & 4 & 5 & 5 \\
    tolerable faults & & & & & & & & & \\
    \hline
    Lower bound of & 17 & 10 & 13 & 10 & 7 & 8 & 7 & 9 & 8 \\
    tolerance (\%) & & & & & & & & & \\
    \hline
  \end{tabular}
\end{table}

Based on the algorithm, we inquire basic questions regarding the scaling laws of system resilience,
such as how many p2p synchronization faults that any $N$-node system can tolerate in that the algorithm will not give wrong estimates of the clock offsets and the p2p clock synchronization faults.
Our analysis gives the lower bound of the number of p2p synchronization faults that any $N$-node system can tolerate when $N$ is up to 12. The result is given in Table~\ref{tab:lower-bound}. If the number of faults is no greater than the lower bound, the faults can be identified and corrected by the algorithm. By defining the tolerance as the ratio between the number of tolerable faults to the total number of p2p synchronization sessions, the third row of Table~\ref{tab:lower-bound} shows the lower bound of the tolerance. Our results can guide the design of network clock synchronization systems with potential p2p synchronization faults. Moreover, we prove that any $N$-node system with $N \ge 3$ cannot tolerate more than $N-2$ p2p synchronization faults.

When the number of faults is greater than the lower bound given in Table~\ref{tab:lower-bound} and no greater than the $N-2$ upper bound, whether the system can tolerate the faults is still an open issue. It is of great interest for future research to explore the tight bound of the fault tolerance.

The remainder of this paper is organized as follows. \sect\ref{sec:related} reviews related work. \sect\ref{sec:problem} introduces the background and states the problem. \sect\ref{sec:analysis} analyzes the resilience bounds.
\sect\ref{sec:conclude} concludes the paper.


\section{Related Work}
\label{sec:related}


Highly stable time sources are often ill-suited for sensor networks. Despite initial study of using chip-scale atomic clock (CSAC) on sensor platforms \cite{dongare2017pulsar}, CSAC is still too expensive (\$1,500 per unit \cite{dongare2017pulsar}) for wide adoption. The Global Positioning System (GPS) and several timekeeping radio stations (e.g., WWVB in U.S.) can provide highly stable global time. However, GPS and radio receivers have various limitations such as high power consumption, poor signal reception in indoor environments (e.g., 47\% good time for WWVB \cite{chen2011ultra}), and susceptibility to wireless spoofing attacks \cite{nighswander2012gps}. Thus, GPS and radio receivers are often used on a limited number of time masters with clear sky views, carefully installed antennas, and sufficient physical air gap to provide global time to a large number of slave nodes via some clock synchronization protocol (e.g., NTP). The resilience of this clock synchronization  protocol between the master and the slaves is the focus of this paper.

Various sensing-based approaches exploit external periodic signals for clock synchronization \cite{viswanathan2016exploiting,rabadi2017taming,yan2017}, time fingerprinting \cite{lukac2009recovering,gupchup2009sundial,viswanathan2016exploiting,li2017natural}, and clock calibration \cite{rowe2009low,li2012flight,hao2011,li2011exploiting}. Time fingerprinting approaches focus on studying the global time information embedded in the sensing data such as microseisms \cite{lukac2009recovering}, sunlight \cite{gupchup2009sundial}, and powerline electromagnetic radiation (EMR) \cite{li2017natural}. They can be a basis for clock synchronization. For instance, the secure clock synchronization approach in \cite{viswanathan2016exploiting} is based on the time fingerprints found in power grid voltage. These studies focus on the p2p synchronization. In this paper, we study the resilience bounds of network clock synchronization against p2p synchronization faults.


Different from clock synchronization that ensures the clocks to have the same value, {\em clock calibration} ensures different clocks to advance at the same speed. The approaches presented in \cite{rowe2009low,li2012flight,hao2011,li2011exploiting} exploit powerline EMR, fluorescent lamp flickering, Wi-Fi beacons, and FM Radio Data System broadcasts to calibrate clocks.
However, clock calibration does not address the resilience issues of clock synchronization. In particular, the sensing-based clock calibration is also prone to faults that can subvert the network clock synchronization.

The resilience of network clock synchronization against Byzantine clock faults has been studied \cite{dolev1984possibility,Lamport1985}. A Byzantine faulty clock gives an arbitrary clock value whenever being read. It has been proved that, to guarantee the synchronization of non-faulty clocks in the presence of $m$ faulty clocks, a total of at least $(3m+1)$ clocks are needed. Different from the Byzantine faulty clock model, we consider faulty p2p synchronization sessions between clocks. The conversion of our problem to the Byzantine clock synchronization problem by considering either node involving a faulty p2p synchronization session as a faulty clock is {\em invalid}, because this faulty clock after the conversion is not a Byzantine faulty clock, unless all p2p synchronization sessions involving this clock are faulty. As our problem does not have this assumption, the resilience bound obtained in \cite{dolev1984possibility,Lamport1985} is not applicable to our problem.

\section{Background and Problem Statement}
\label{sec:problem}

\subsection{Background and Preliminaries}
\label{subsec:background}

\begin{figure}
  \centering
  \includegraphics{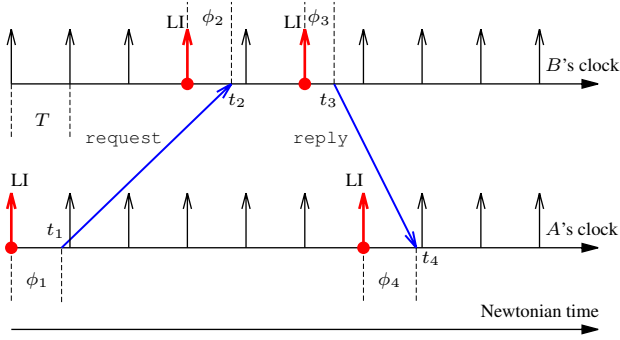}
  \caption{Principle of sensing-based p2p clock synchronization.}
  \label{fig:principle}
\end{figure}


\subsubsection{Sensing-based p2p clock synchronization}
This section describes the principle of the sensing-based p2p clock synchronization that exploits external periodic signals. Without loss of generality, we assume that the external periodic signals sensed by the two peers, nodes $A$ and $B$, are two synchronous Dirac combs with the same period $T$. Fig.~\ref{fig:principle} illustrates the two Dirac combs in the same Newtonian time frame. The objective of the sensing-based p2p clock synchronization is to estimate the offset between $A$'s and $B$'s clocks by using the Dirac combs.

To simplify the analysis of the clock offset estimation, we assume that $A$'s and $B$'s clocks advance at the same speed, such that the offset between the two clocks is a constant within a concerned time period before any clock is adjusted according to the estimated clock offset to achieve clock synchronization. In existing sensing-based p2p clock synchronization approaches \cite{viswanathan2016exploiting,rabadi2017taming,yan2017}, a {\em synchronization session}, i.e., the process of estimating the clock offset, takes a short time (e.g., tens of milliseconds in \cite{rabadi2017taming}). Typical crystal oscillators found in microcontrollers and personal computers have drift rates of 30 to 50 parts-per-million (ppm) \cite{hao2011}. Thus, the change of the clock offset during a synchronization session of 100 milliseconds is at most 5 microseconds only, whereas the clock offset estimation errors of successful synchronization sessions are at sub-millisecond \cite{viswanathan2016exploiting,rabadi2017taming} or milliseconds levels \cite{yan2017} in practice. Thus, the clock offset estimation errors caused by signal noises are much larger than those caused by the two peers' different clock speeds.


\subsubsection{Fault model}
A synchronization session is {\em successful} (or {\em non-faulty}) if it identifies the correspondence between an $A$'s Dirac impulse and a $B$'s Dirac impulse that occur at the same Newtonian time instant; otherwise, the synchronization session is {\em faulty}. Since the two Dirac combs are synchronous, a successful synchronization session gives a zero clock offset estimation error, whereas a faulty synchronization session gives a clock offset estimation error of $nT$, where $n$ is a non-zero integer.

\subsubsection{Other related issues}
It has been shown in \cite{viswanathan2016exploiting,rabadi2017taming}, if the Dirac combs are practically difficult for the attacker to tamper with or jam, the sensing-based p2p clock synchronization can address the packet delay attack, which is an open issue that cannot be solved by conventional security measures \cite{rfc7384,mizrahi2012game,ullmann2009delay}. However, due to process noises of the external signals and sensor hardware noises/faults, the sensing-based p2p synchronization can be faulty. For self-containment of this paper, Appendix~\ref{appendix:reasons} reviews the detailed reasons of the faults. In this paper, we focus on the fault tolerance of sensing-based synchronization. Built upon the secure p2p synchronization \cite{viswanathan2016exploiting,rabadi2017taming}, the clock synchronization approach presented in this paper is resilient against both the packet delay attacks and synchronization faults.

We note that, in practice, the two Dirac combs may not be perfectly synchronous. For instance, in \cite{viswanathan2016exploiting,rabadi2017taming}, the time displacement between the two Dirac combs is about 0.05\% to 0.5\% of $T$. This time displacement is the major source of the clock offset estimation error. As the time displacements are much smaller than the synchronization faults (i.e., $nT$),
we can easily classify successful and faulty synchronization sessions by comparing the clock offset estimation error with a threshold (e.g., $\frac{T}{2}$).
For simplicity of exposition, we ignore the time displacement in our analysis regarding the system resilience against faulty synchronization sessions.

\subsection{Network Clock Synchronization}

To improve the robustness of clock synchronization against p2p synchronization faults, this section proposes an approach to cross-check the p2p synchronization results among multiple nodes and correct the faults if present.


Consider a system of $N$ nodes: $\{n_0, n_1,...n_{N-1}\}$. Let $\delta_{ij}$ denote the offset between the clocks of $n_i$ and $n_j$, which is unknown and to be estimated. Specifically, $\delta_{ij} = c_i(t) - c_j(t)$, where $c_i(t)$ and $c_j(t)$ are the clock values of $n_i$ and $n_j$ at any given time instant $t$, respectively. As discussed in \sect\ref{subsec:background}, we assume that $\delta_{ij}$ is time-invariant. By designating $n_0$ as the reference node, we have $\delta_{ij} = \delta_{i0} - \delta_{j0}$. Any pair of two nodes, $n_i$ and $n_j$, will perform a synchronization session using the sensing-based p2p clock synchronization to measure $\delta_{ij}$. Denote by $n_i \leftrightarrow n_j$ the synchronization session between $n_i$ and $n_j$. Denote by $\widetilde{\delta}_{ij}$ the measured clock offset. If the synchronization session is successful, $\widetilde{\delta}_{ij} = \delta_{ij}$; if the synchronization session is faulty, $\widetilde{\delta}_{ij} = \delta_{ij} + e_{ij}$, where $e_{ij}$ is the p2p synchronization fault. Every node pair performs a p2p synchronization session. Thus, there will be a total of ${N \choose 2} = \frac{N(N-1)}{2}$ p2p synchronization sessions.

All the $\frac{N(N-1)}{2}$ clock offset measurements are transmitted to a central node, which runs a fault-tolerate network clock synchronization algorithm. Denote by $\hat{\delta}_{ij}$ and $\hat{e}_{ij}$ the estimates for $\delta_{ij}$ and $e_{ij}$, respectively. A general equation system assuming all the p2p synchronization sessions are faulty is
\begin{equation}
\left\{
\begin{array}{ll}
\hat{\delta}_{j0} + \hat{e}_{j0} = \widetilde{\delta}_{j0}, & \forall j \in [1, N-1]; \\
\hat{\delta}_{i0} - \hat{\delta}_{j0} + \hat{e}_{ij} = \widetilde{\delta}_{ij}, & \forall i, j \in [1, N-1], i > j.
\end{array}
\right.
\label{eq:error-correction}
\end{equation}
The variables to be solved are the unknowns $\{ \hat{\delta}_{j0}| \forall j \in [1, N-1] \}$ and $\{\hat{e}_{ij} | \forall i,j \in [0, N-1], i > j \}$, where $\hat{\delta}_{j0}$ is the estimated clock offset between $n_j$ and the reference node $n_0$; $\hat{e}_{ij}$ is the estimated p2p clock synchronization fault between $n_i$ and $n_j$.


If the network clock synchronization algorithm considers that a total of $k$ p2p synchronization sessions are faulty, it keeps $k$ estimated p2p synchronization faults (i.e., $\hat{e}_{ij}$) in Eq.~(\ref{eq:error-correction}) and removes other estimated p2p synchronization faults. Thus, there will be ${\frac{N(N-1)}{2} \choose k}$ possible distributions of the $k$ estimated p2p synchronization faults among a total of $\frac{N(N-1)}{2}$ p2p synchronization sessions. Algorithm~\ref{alg:error-correction} shows the pseudocode of algorithm. It starts by assuming there are no faults (i.e., $k=0$). In each iteration that increases $k$ by one, it solves Eq.~(\ref{eq:error-correction}) for all possible distributions of the $k$ estimated p2p synchronization faults.
Once a solution is found, Algorithm~\ref{alg:error-correction} returns.

\renewcommand{\algorithmicrequire}{\textbf{Given:}}
\renewcommand{\algorithmicensure}{\textbf{Output:}}
\renewcommand{\algorithmiccomment}[1]{// #1}

\begin{algorithm}[t]
\caption{Fault-tolerate network clock synchronization.}
\label{alg:error-correction}
\small
\begin{algorithmic}[1]
\REQUIRE $\{ \widetilde{\delta}_{ij} | \forall i,j \in [0, N-1], i > j \}$
\ENSURE $\{ \hat{\delta}_{ij}, \hat{e}_{ij} | \forall i,j \in [0, N-1],i > j \}$

\STATE $k=0$
\WHILE{$k \le \frac{N(N-1)}{2}$}
\FOR{each distribution of the $k$ estimated p2p synchronization faults among the $\frac{N(N-1)}{2}$ p2p synchronization sessions}
\label{line:foreach}
\IF{the corresponding Eq.~(\ref{eq:error-correction}) has a solution}
\label{line:solve}
\STATE return $\{ \hat{\delta}_{ij}, \hat{e}_{ij} | \forall i,j \in [0, N-1], i > j \}$
\ENDIF
\ENDFOR
\label{line:endforeach}
\STATE $k = k + 1$
\ENDWHILE
\end{algorithmic}
\end{algorithm}

Algorithm~\ref{alg:error-correction} requires neither the number nor the distribution of the actual p2p synchronization faults. Whether it can correct the faults and how many faults it can tolerate will be the focus of this paper. Algorithm~\ref{alg:error-correction} is executed on a central node; its fault tolerance performance, which is the focus of this paper, will provide important understanding.




\subsection{Problem Statement}

\begin{definition}[$K$-resilience]
  Let $K \in \mathbb{Z}_{\ge 0}$ denote the number of faulty p2p synchronization sessions among a total of $\frac{N(N-1)}{2}$ sessions in an $N$-node system. The system with Algorithm ~\ref{alg:error-correction} is $K$-resilient if the algorithm can correct any $K$ non-zero p2p synchronization faults.\qed
  \label{def:resilience}
\end{definition}

From Algorithm~\ref{alg:error-correction}, we define the $K$-resilience condition that can be used to check whether a system is $K$-resilient.

\begin{definition}[$K$-resilience condition]
A system with Algorithm~\ref{alg:error-correction} is $K$-resilient if the following conditions are satisfied:
\begin{enumerate}
\item $\forall k \in [0,K)$, Eq.~(\ref{eq:error-correction}) constructed with any distribution of the $K$ actual p2p synchronization faults and any distribution of the $k$ estimated p2p synchronization faults has no solutions;
\item When $k=K$, for any distribution of the $K$ actual p2p synchronization faults and any distribution of the $k$ estimated p2p synchronization faults,
  \begin{enumerate}
  \item if the distribution of the $k$ estimated p2p synchronization faults is identical to the distribution of the actual faults, Eq.~(\ref{eq:error-correction}) has a unique solution;
  \item otherwise, Eq.~(\ref{eq:error-correction}) has no solutions. \qed
  \end{enumerate}
\end{enumerate}
\label{def:resilience-condition}
\end{definition}

Note that in the condition 2)-a) of Definition~\ref{def:resilience-condition}, the unique solution must give the correct estimates of the clock offsets and the p2p synchronization faults.

We aim at analyzing the following resilience bounds:

\begin{definition}[Lower bound of maximum resilience]
  A function $f_l(N)$ is a lower bound of maximum resilience if any $N$-node system with Algorithm~\ref{alg:error-correction} is $K$-resilient for $K \le f_l(N)$.
\end{definition}

\begin{definition}[Upper bound of maximum resilience]
  A function $f_u(N)$ is a upper bound of maximum resilience if any $N$-node system with Algorithm~\ref{alg:error-correction} is not $K$-resilient for $K > f_u(N)$.
\end{definition}

\begin{definition}[Tight bound of maximum resilience]
  A function $f_t(N)$ is a tight bound of maximum resilience if any $N$-node system with Algorithm~\ref{alg:error-correction} is $K$-resilient for $K \le f_t(N)$ and not $K$-resilient for $K > f_t(N)$.
\end{definition}



\section{Vectorization and $K$-Resilience}
\label{sec:analysis}

\subsection{Vectorization}

We vectorize the representation of Eq.~(\ref{eq:error-correction}) that is solved by Line~\ref{line:solve} of Algorithm~\ref{alg:error-correction}. Define $\hat{\vecg{\delta}} \in \mathbb{R}^{N-1}$ composed of all clock offset estimates, i.e., $\hat{\vecg{\delta}} = \left(\hat{\delta}_{10}, \hat{\delta}_{20}, \ldots, \hat{\delta}_{(N-1)0}\right)^\intercal$. Define $\hat{\vec{e}} \in \mathbb{R}^k$ composed of the $k$ p2p synchronization fault estimates. Eq.~(\ref{eq:error-correction}) can be rewritten as $
  \left( \mat{A}_1 \mat{A}_2 \right)
  \left(
  \begin{array}{c}
    \hat{\vecg{\delta}} \\
    \hat{\vec{e}}
  \end{array}
\right) = \vec{b}
$,
where $\mat{A}_1 \in \mathbb{R}^{\frac{N(N-1)}{2} \times (N-1)}$ and $\mat{A}_2 \in \mathbb{R}^{\frac{N(N-1)}{2} \times k}$ are two matrices composed of -1, 0, and 1 containing coefficients corresponding to $\hat{\delta}_{\cdot 0}$ and $\hat{e}_{\cdot\cdot}$, respectively;
the vector $\vec{b} \in \mathbb{R}^{\frac{N(N-1)}{2}}$ consists of all the measured clock offsets. To simplify notation, we define $\mat{A} = \left( \mat{A}_1 \mat{A}_2 \right)$ and $\vec{x} = \left( \begin{array}{c}
    \hat{\vecg{\delta}} \\
    \hat{\vec{e}}
  \end{array} \right)$. From the Rouch\'{e}-Capelli theorem \cite{Shafarevich2012}, the necessary and sufficient condition that $\mat{A}\vec{x} = \vec{b}$ has no solutions is $\mathrm{rank}(\mat{A} | \vec{b}) \neq \mathrm{rank}(\mat{A})$, where $\mat{A} | \vec{b}$ is the augmented matrix.

\subsection{$K$-Resilience under Certain Settings}
\label{subsec:manual-check}

This section presents the analysis on the $K$-resilience of an $N$-node system with Algorithm~\ref{alg:error-correction} under certain settings of $K$ and $N$. This analysis provides insights into the more general analysis of the lower/upper bounds of maximum resilience.

\begin{proposition}
  A 3-node system is not 1-resilient.
  \label{exmp:1}
\end{proposition}
\begin{proof}
  Consider a case where the p2p synchronization session $n_1 \leftrightarrow n_2$ is faulty. When $k=0$ in Algorithm~\ref{alg:error-correction}, the vectorized equation system in Eq.~(\ref{eq:error-correction}) is
  \begin{equation*}
    \begin{array}{cccc}
      \left(
        \begin{array}{cc}
          1 & 0 \\
          0 & 1 \\
          -1 & 1
        \end{array}
      \right)
      &
      \left(
        \begin{array}{c}
          \hat{\delta}_{10} \\
          \hat{\delta}_{20}
        \end{array}
      \right)
      &
      =
      &
      \left(
        \begin{array}{c}
          \delta_{10} \\
          \delta_{20} \\
          \delta_{20} - \delta_{10} + e_{21}
        \end{array}
      \right).
      \\
      \Uparrow & \Uparrow & & \Uparrow \\
      \mat{A} & \vec{x} &  & \vec{b}
    \end{array}
  \end{equation*}
  Note that $\mat{A}_2$ and $\hat{\vec{e}}$ are empty. With $e_{21} \neq 0$, Gaussian elimination shows that $\mathrm{rank}(\mat{A} | \vec{b}) \neq \mathrm{rank}(\mat{A})$. Thus, the equation system has no solutions and Algorithm~\ref{alg:error-correction} will move on to the case of $k=1$. The algorithm will attempt to test all the $\frac{N(N-1)}{2}=3$ possible cases of a single faulty p2p synchronization session. For instance, when the algorithm assumes that $n_0 \leftrightarrow n_1$ is faulty, the equation system is
    \begin{equation*}
    \left(
      \begin{array}{ccc}
        1 & 0 & 1\\
        0 & 1 & 0 \\
        -1 & 1 & 0
      \end{array}
    \right)
    \left(
      \begin{array}{c}
        \hat{\delta}_{10} \\
        \hat{\delta}_{20} \\
        \hat{e}_{10}
      \end{array}
    \right)
    =
    \left(
      \begin{array}{c}
        \delta_{10} \\
        \delta_{20} \\
        \delta_{20} - \delta_{10} + e_{21}
      \end{array}
    \right).
  \end{equation*}
  With $e_{21} \neq 0$, we have $\mathrm{rank}(\mat{A}|\vec{b}) = \mathrm{rank}(\mat{A})$ and $\mat{A}$ has full column rank. Thus, the equation system has a unique solution. Therefore, the condition 2)-b) of Definition~\ref{def:resilience-condition} is not satisfied and the 3-node system is not 1-resilient. In fact, the unique solution must be a wrong solution, which is $\{ \hat{\delta}_{10} = \delta_{10} - e_{21}, \hat{\delta}_{20} = \delta_{20}, \hat{e}_{10} = e_{21} \}$.
\end{proof}

\begin{proposition}
  A 4-node system is 1-resilient.
  \label{exmp:2}
\end{proposition}

We provide a sketch of the proof as follows instead of the complete proof due to space limit. Consider a case where the p2p synchronization session $n_0 \leftrightarrow n_2$ is faulty. When $k=0$ in Algorithm~\ref{alg:error-correction},
similar to Proposition~\ref{exmp:1}, the equation system has no solutions and Algorithm~\ref{alg:error-correction} will move on to the case of $k=1$. The algorithm will test all the $\frac{N(N-1)}{2}=6$ possible cases of a single faulty p2p synchronization session. For instance, when the algorithm assumes $n_0 \leftrightarrow n_1$ is faulty, the vectorized equation system is
\begin{equation}
\left(
\begin{array}{cccc}
1 & 0 & 0 & 1 \\
0 & 1 & 0 & 0 \\
0 & 0 & 1 & 0 \\
-1 & 1 & 0 & 0 \\
0 & -1 & 1 & 0 \\
-1 & 0 & 1 & 0
\end{array}
\right)
\left(
\begin{array}{c}
\hat{\delta}_{10}\\
\hat{\delta}_{20}\\
\hat{\delta}_{30}\\
\hat{e}_{10}
\end{array}
\right)
=
\left(
\begin{array}{c}
\delta_{10} \\
\delta_{20} + e_{20} \\
\delta_{30} \\
\delta_{20} - \delta_{10} \\
\delta_{30} - \delta_{20} \\
\delta_{30} - \delta_{10}
\end{array}
\right).
\label{eq:4node-system0}
\end{equation}
As $\mathrm{rank}(\mat{A} | \vec{b}) \neq \mathrm{rank}(\mat{A})$, the equation system has no solutions. An exhaustive check shows that, only when the algorithm assumes the synchronization session between $n_0$ and $n_2$ is faulty, the equation system has a unique solution (i.e., $\mathrm{rank}(\mat{A} | \vec{b}) = \mathrm{rank}(\mat{A})$ and $\mat{A}$ has full column rank). Thus, the algorithm can correct the fault. In fact, it can be verified that, for the 4-node system, no matter which p2p synchronization session is faulty, the algorithm can correct the fault. Therefore, the 4-node system is 1-resilient.

\begin{proposition}
  A 4-node system is not 2-resilient.
  \label{exmp:3}
\end{proposition}
\begin{proof}
  Consider the 4-node system with two faulty p2p synchronization sessions: $n_0 \leftrightarrow n_1$ and $n_0 \leftrightarrow n_2$. When $k=0$, the equation system has no solutions. When $k=1$, consider a case where $n_0 \leftrightarrow n_3$ is assumed to be faulty by the algorithm. The vectorized equation system is
  \begin{equation}
    \left(
      \begin{array}{cccc}
        1 & 0 & 0 & 0 \\
        0 & 1 & 0 & 0 \\
        0 & 0 & 1 & 1 \\
        -1 & 1 & 0 & 0 \\
        -1 & 0 & 1 & 0 \\
        0 & -1 & 1 & 0
      \end{array}
    \right)
    \left(
      \begin{array}{c}
        \hat{\delta}_{10} \\
        \hat{\delta}_{20} \\
        \hat{\delta}_{30} \\
        \hat{e}_{30}
      \end{array}      
    \right)
    =
    \left(
      \begin{array}{c}
        \delta_{10} + e_{10} \\
        \delta_{20} + e_{20} \\
        \delta_{30} \\
        \delta_{20} - \delta_{10} \\
        \delta_{30} - \delta_{10} \\
        \delta_{30} - \delta_{20}
      \end{array}
    \right).
    \label{eq:4node-system1}
  \end{equation}
  If $e_{10} \neq e_{20}$, $\mathrm{rank}(\mat{A}|\vec{b}) \neq \mathrm{rank}(\mat{A})$ and the equation system has no solutions. However, if $e_{10} = e_{20}$, $\mathrm{rank}(\mat{A} | \vec{b}) = \mathrm{rank}(\mat{A})$ and $\mat{A}$ has full column rank; the equation system has a unique wrong solution of $\{\hat{\delta}_{10} = \delta_{10} + e_{10}, \hat{\delta}_{20} = \delta_{20} + e_{10},  \hat{\delta}_{30} = \delta_{30} + e_{10}, \hat{e}_{30} = - e_{10}\}$. Although this counterexample against the 4-node system's 2-resilience is obtained under a certain condition of $e_{10} = e_{20}$, we can conclude that the 4-node system is not 2-resilient.
  \end{proof}

  To gain more insights, we also analyze a case of $k=2$ with $n_0 \leftrightarrow n_1$ and $n_0 \leftrightarrow n_3$ assumed to be faulty by the algorithm. The vectorized equation system is
  \begin{equation}
    \small
      \left(
      \begin{array}{ccccc}
        1 & 0 & 0 & 1 & 0 \\
        0 & 1 & 0 & 0 & 0 \\
        0 & 0 & 1 & 0 & 1 \\
        -1 & 1 & 0 & 0 & 0 \\
        -1 & 0 & 1 & 0 & 0 \\
        0 & -1 & 1 & 0 & 0
      \end{array}
    \right)
    \left(
      \begin{array}{c}
        \hat{\delta}_{10} \\
        \hat{\delta}_{20} \\
        \hat{\delta}_{30} \\
        \hat{e}_{10} \\
        \hat{e}_{30}
      \end{array}      
    \right)
    \!\!=\!\!
    \left(
      \begin{array}{c}
        \delta_{10} + e_{10} \\
        \delta_{20} + e_{20} \\
        \delta_{30} \\
        \delta_{20} - \delta_{10} \\
        \delta_{30} - \delta_{10} \\
        \delta_{30} - \delta_{20}
      \end{array}
    \right).
    \label{eq:4node-system2}
\end{equation}
As $\mathrm{rank}(\mat{A} | \vec{b}) = \mathrm{rank}(\mat{A})$ and $\mat{A}$ has full column rank, the equation system has a unique solution, which violates the 2-resilience condition. In fact, the equation system has a unique wrong solution that does not require any relationship between $e_{10}$ and $e_{20}$: $\{\hat{\delta}_{10} = \delta_{10} + e_{20}, \hat{\delta}_{20} = \delta_{20} + e_{20}, \hat{\delta}_{30} = \delta_{30} + e_{20}, \hat{e}_{10} = e_{10} - e_{20}, \hat{e}_{30} = -e_{20}\}$.

\begin{proposition}
  A 5-node system is 1-resilient.
  \label{exmp:4}
\end{proposition}

We provide a sketch of the proof as follows instead of the complete proof due to space limit. Consider a 5-node system with one p2p synchronization fault. The resilience is independent from how we name the nodes. We name the two involving nodes of the faulty synchronization session to be $n_0$ and $n_1$. An exhaustive check over all the $5 \choose 2$ possible cases for a single assumed faulty synchronization session shows that the 1-resilience condition is satisfied. Thus, the 5-node system is 1-resilient.

\begin{proposition}
  A 5-node system is not 2-resilient.
  \label{exmp:5}
\end{proposition}
\begin{proof}
  We consider a 5-node system, in which (i) the p2p synchronization sessions $n_0 \leftrightarrow n_1$ and $n_1 \leftrightarrow n_4$ are faulty and (ii) the p2p synchronization sessions $n_1 \leftrightarrow n_2$ and $n_1 \leftrightarrow n_3$ are assumed by the algorithm to be faulty. The vectorized equation system is
  \begin{equation}
    \scriptsize
    \left(
      \begin{array}{cccccc}
        1 & 0 & 0 & 0 & 0 & 0 \\
        0 & 1 & 0 & 0 & 0 & 0 \\
        0 & 0 & 1 & 0 & 0 & 0 \\
        0 & 0 & 0 & 1 & 0 & 0 \\
        -\!1 & 1 & 0 & 0 & 1 & 0 \\
        -\!1 & 0 & 1 & 0 & 0 & 1 \\
        -\!1 & 0 & 0 & 1 & 0 & 0 \\
        0 & -\!1 & 1 & 0 & 0 & 0 \\
        0 & -\!1 & 0 & 1 & 0 & 0 \\
        0 & 0 & -\!1 & 1 & 0 & 0        
      \end{array}
    \right)
    \!\!
    \left(
      \begin{array}{c}
        \hat{\delta}_{10} \\
        \hat{\delta}_{20} \\
        \hat{\delta}_{30} \\
        \hat{\delta}_{40} \\
        \hat{e}_{21} \\
        \hat{e}_{31}
      \end{array}
    \right)
    \!\!=\!\!
    \left(
      \begin{array}{c}
        \delta_{10} + e_{10} \\
        \delta_{20} \\
        \delta_{30} \\
        \delta_{40} \\
        \delta_{20} - \delta_{10} \\
        \delta_{30} - \delta_{10} \\
        \delta_{40} \!-\! \delta_{10} \!+\! e_{41} \\
        \delta_{30} - \delta_{20} \\
        \delta_{40} - \delta_{20} \\
        \delta_{40} - \delta_{30}
      \end{array}
    \right)\!\!.
    \label{eq:5node-system}
  \end{equation}
If $e_{10} = -e_{41}$, the equation system has a unique solution of $\{\hat{\delta}_{10} = \delta_{10} + e_{10}, \hat{\delta}_{20} = \delta_{20}, \hat{\delta}_{30} = \delta_{30}, \hat{\delta}_{40} = \delta_{40}, \hat{e}_{21} = e_{10}, \hat{e}_{31} = e_{10} \}$, which violates the resilience condition. Thus, a 5-node system is not 2-resilient.
\end{proof}

\subsection{Re-Vectorization}

In \sect\ref{subsec:manual-check}, we adopt an approach of enumerating counterexamples to prove that a system is not $K$-resilient. As shown in the proofs of Propositions~\ref{exmp:3} and \ref{exmp:5}, if the actual faults satisfy certain conditions, the rank of $\mat{A}|\vec{b}$ may change, presenting a pitfall to the approach of enumerating counterexamples. This motivates us to consider the actual faults as the variables of the equation system in Eq.~(\ref{eq:error-correction}). The following re-vectorization will be used in \sect\ref{subsec:symbolic} to derive the lower bound of maximum resilience.

By defining a vector $\vec{e} \in \mathbb{R}^{K}$ composed of the $K$ actual p2p synchronization faults, we can reformat $\mat{A}\vec{x} = \vec{b}$ to include the actual faults into the vector of unknowns:
\begin{equation}
  \mat{A}' \vec{x}' = \vec{b}', \text{where } \vec{x}' = \left( \begin{array}{c} \hat{\vecg{\delta}} \\ \hat{\vec{e}} \\ \vec{e} \end{array} \right), \mat{A}' = \left( \mat{A}_1 \mat{A}_2 \mat{A}_3 \right),
\end{equation}
$\mat{A}_3 \in \mathbb{R}^{\frac{N(N-1)}{2} \times K}$ is a matrix corresponding to $\vec{e}$, $\vec{b}' \in \mathbb{R}^{\frac{N(N-1)}{2}}$ consists of the actual clock offsets.

The re-vectorization of the equation systems in Eqs.~(\ref{eq:4node-system0}), (\ref{eq:4node-system1}), and (\ref{eq:4node-system2}) are respectively given by
\begin{equation}
  \small
  \left(
    \begin{array}{ccccc}
      1 & 0 & 0 & 1 & 0 \\
      0 & 1 & 0 & 0 & -1 \\
      0 & 0 & 1 & 0 & 0 \\
      -1 & 1 & 0 & 0 & 0 \\
      0 & -1 & 1 & 0 & 0 \\
      -1 & 0 & 1 & 0 & 0
    \end{array}
  \right)
  \!\!
  \left(
    \begin{array}{c}
      \hat{\delta}_{10}\\
      \hat{\delta}_{20}\\
      \hat{\delta}_{30}\\
      \hat{e}_{10} \\
      e_{20}
    \end{array}
  \right)
  \!=\!
  \left(
    \begin{array}{c}
      \delta_{10} \\
      \delta_{20} \\
      \delta_{30} \\
      \delta_{20} - \delta_{10} \\
      \delta_{30} - \delta_{20} \\
      \delta_{30} - \delta_{10}
    \end{array}
  \right),
  \label{eq:4node-system0-reformat}
\end{equation}

\begin{equation}
  \footnotesize
    \left(
      \begin{array}{cccccc}
        1 & 0 & 0 & 0 & -\!1 & 0 \\
        0 & 1 & 0 & 0 & 0 & -1 \\
        0 & 0 & 1 & 1 & 0 & 0 \\
        -1 & 1 & 0 & 0 & 0 & 0 \\
        -1 & 0 & 1 & 0 & 0 & 0 \\
        0 & -1 & 1 & 0 & 0 & 0
      \end{array}
    \right)
    \!\!
    \left(
      \begin{array}{c}
        \hat{\delta}_{10} \\
        \hat{\delta}_{20} \\
        \hat{\delta}_{30} \\
        \hat{e}_{30} \\
        e_{10} \\
        e_{20}
      \end{array}      
    \right)
    \!\!=\!\!
    \left(
      \begin{array}{c}
        \delta_{10} \\
        \delta_{20} \\
        \delta_{30} \\
        \delta_{20} \!-\! \delta_{10} \\
        \delta_{30} \!-\! \delta_{10} \\
        \delta_{30} \!-\! \delta_{20}
      \end{array}
    \right),
    \label{eq:4node-system1-reformat}
  \end{equation}

  \begin{equation}
    \scriptsize
      \left(
      \begin{array}{ccccccc}
        1 & 0 & 0 & 1 & 0 & -\!1 & 0 \\
        0 & 1 & 0 & 0 & 0 & 0 & -\!1 \\
        0 & 0 & 1 & 0 & 1 & 0 & 0 \\
        -\!1 & 1 & 0 & 0 & 0 & 0 & 0 \\
        -\!1 & 0 & 1 & 0 & 0 & 0 & 0 \\
        0 & -\!1 & 1 & 0 & 0 & 0 & 0
      \end{array}
    \right)
    \!\!
    \left(
      \begin{array}{c}
        \hat{\delta}_{10} \\
        \hat{\delta}_{20} \\
        \hat{\delta}_{30} \\
        \hat{e}_{10} \\
        \hat{e}_{30} \\
        e_{10} \\
        e_{20}
      \end{array}      
    \right)
    \!\!=\!\!
    \left(
      \begin{array}{c}
        \delta_{10} \\
        \delta_{20} \\
        \delta_{30} \\
        \delta_{20} \!-\! \delta_{10} \\
        \delta_{30} \!-\! \delta_{10} \\
        \delta_{30} \!-\! \delta_{20}
      \end{array}
    \right)\!\!,
    \label{eq:4node-system2-reformat}
  \end{equation}


  In Eq.~(\ref{eq:4node-system0-reformat}), $\mathrm{rank}(\mat{A}' | \vec{b}) = \mathrm{rank}(\mat{A}')$ and $\mat{A}'$ has full column rank. Thus, Eq.~(\ref{eq:4node-system0-reformat}) has a unique solution, which is $\{\hat{\delta}_{10} = \delta_{10}, \hat{\delta}_{20} = \delta_{20}, \hat{\delta}_{30} = \delta_{30}, \hat{e}_{10} = 0, e_{20} = 0 \}$. This is consistent with the observation in the proof sketch of Proposition~\ref{exmp:2} that $\mat{A}\vec{x} = \vec{b}$ has no solutions if $e_{20} \neq 0$.

  In Eq.~(\ref{eq:4node-system1-reformat}), $\mathrm{rank}(\mat{A}' | \vec{b}) = \mathrm{rank}(\mat{A}')$ and $\mat{A}'$ is not full column ranked. Thus, $\mat{A}'\vec{x}' = \vec{b}'$ has an infinite number of solutions. Applying Guassian elimination to Eq.~(\ref{eq:4node-system1-reformat}) gives $\{\hat{\delta}_{10} = \delta_{10} + e_{10}, \hat{\delta}_{20} = \delta_{20} + e_{10},  \hat{\delta}_{30} = \delta_{30} + e_{10}, \hat{e}_{30} = - e_{10}, e_{20} = e_{10} \}$, where $e_{10}$ and $e_{20}$ are considered as variables in $\mat{A}'\vec{x}' = \vec{b}'$, not as constants in $\mat{A}\vec{x} = \vec{b}$. The above result means that there exist non-zero $e_{10}$ and $e_{20}$ such that the solution of $\mat{A}\vec{x} = \vec{b}$ is wrong.

  In Eq.~(\ref{eq:4node-system2-reformat}), $\mathrm{rank}(\mat{A}' | \vec{b}) = \mathrm{rank}(\mat{A}')$ and $\mat{A}'$ is not full column ranked. Thus, $\mat{A}'\vec{x}' = \vec{b}'$ has an infinite number of solutions. Applying Gaussian elimination to Eq.~(\ref{eq:4node-system2-reformat}) gives the relationship derived in the proof of Proposition~\ref{exmp:3}, i.e., $\{\hat{\delta}_{10} = \delta_{10} + e_{20}, \hat{\delta}_{20} = \delta_{20} + e_{20}, \hat{\delta}_{30} = \delta_{30} + e_{20}, \hat{e}_{10} = e_{10} - e_{20}, \hat{e}_{30} = -e_{20}\}$, where $e_{10}$ and $e_{20}$ are considered as variables in $\mat{A}'\vec{x}' = \vec{b}'$, not as constants in $\mat{A}\vec{x} = \vec{b}$. The above result also shows that there exist non-zero $e_{10}$ and $e_{20}$ such that the solution of $\mat{A}\vec{x} = \vec{b}$ is wrong.

  From the above examples, we can see that the solution to re-vectorization captures the condition that the actual faults need to satisfy such that the $\mat{A}\vec{x}=\vec{b}$ will give wrong solutions.



\section{Bounds of Maximum Resilience}

\subsection{Lower Bound of Maximum Resilience}
\label{subsec:symbolic}


In this section, we first develop two lemmas, Lemma~\ref{lem:1} and Lemma~\ref{lem:2}. The proof of Lemma~\ref{lem:2} uses Lemma~\ref{lem:1}. Then, we prove Proposition~\ref{prop:sufficient} using Lemma~\ref{lem:2}. Proposition~\ref{prop:sufficient} gives a sufficient condition that a system is $K$-resilient. This condition can be used to compute the lower bound of maximum resilience for any $N$-node system.

  \begin{lemma}
    $\mat{A}' \vec{x}' = \vec{b}'$ always has one or more solutions. When $\mat{A}'$ has full column rank, the original $\mat{A} \vec{x} = \vec{b}$ either has no solutions or has a unique correct solution.
    \label{lem:1}
  \end{lemma}
  \begin{proof}
    The $\vec{x}'$ satisfying (i) $\hat{\delta}_{j0} = \delta_{j0}$, $\forall j \in [1, N-1]$, (ii) $\hat{\vec{e}} = \vec{0}$, and (iii) $\vec{e}=\vec{0}$ must be a solution. We denote this solution as $\vec{x}_0'$. As shown in previous examples, $\mat{A}'\vec{x}'=\vec{b}'$ can have an infinite number of solutions. Therefore, $\mathrm{rank}(\mat{A}' | \vec{b}') = \mathrm{rank}(\mat{A}')$ always holds and $\mat{A}'\vec{x}' = \vec{b}'$ always has one or more solutions.

    When $\mat{A}'$ has full column rank, $\mat{A}'\vec{x}'=\vec{b}'$ has a unique solution that must be $\vec{x}_0'$. The $\vec{e}=\vec{0}$ in this solution means that the original $\mat{A}\vec{x} = \vec{b}$ does not allow any p2p synchronization fault. We now consider two cases. First, in the presence of any p2p synchronization fault, the $\mat{A}\vec{x} = \vec{b}$ must have no solutions; otherwise, the solution of $\mat{A}\vec{x} = \vec{b}$ conflicts with the unique solution of $\vec{A}'\vec{x}'=\vec{b}'$ with $\vec{e}=\vec{0}$. Second, in the absence of synchronization fault, the unique solution $\vec{x}_0'$ encompasses the unique correct solution of $\mat{A}\vec{x}=\vec{b}$.
  \end{proof}

We say that an estimated p2p synchronization fault is correctly positioned if the corresponding p2p synchronization session is truly faulty. For example, in Eq.~(\ref{eq:4node-system2-reformat}), the $\hat{e}_{10}$ is correctly positioned, but the $\hat{e}_{30}$ is not correctly positioned.

\begin{lemma}
  When $\mathrm{rank}(\mat{A}') = N - 1 + k + K - l$, where $l \in [0, k]$ is the number of correctly positioned estimated p2p synchronization faults, the original $\mat{A}\vec{x} = \vec{b}$ either has no solutions or has a unique correct solution.
  \label{lem:2}
\end{lemma}

\begin{proof}
  We define three sets: (1) $\mathcal{E}$ is the set of the subscripts of the estimated p2p synchronization faults, (2) $\mathcal{A}$ is the set of the subscripts of the actual p2p synchronization faults, (3) $\mathcal{C}$ is the set of the subscripts of the correctly positioned estimated p2p synchronization faults.

  When $l = 0$, the given condition $\mathrm{rank}(\mat{A}')=N-1+k+K-l$ ensures that $\mat{A}'$ has full column rank. From Lemma~\ref{lem:1}, $\mat{A}\vec{x} = \vec{b}$ has either no solutions or a unique correct solution.

The rest of the proof considers $l \in (0, k]$. We now prove that the $\mathcal{S} = \{\hat{\delta}_{i0} = \delta_{i0}, \hat{e}_{mn} = e_{mn}, \hat{e}_{pq} = 0, e_{xy} = 0 | \forall i \in [1, N-1], \forall mn \in \mathcal{C}, \forall pq \in \mathcal{E} \setminus \mathcal{C}, \forall xy \in \mathcal{A} \setminus \mathcal{C}\}$ is the entire solution space of $\mat{A}'\vec{x}'=\vec{b}'$. First, clearly, $\mathcal{S}$ is a solution subspace of $\mat{A}'\vec{x}' = \vec{b}'$, because it is the correct solution to a system with $l$ actual non-zero p2p synchronization faults and correct distribution of the estimated p2p synchronization faults. The dimension of $\mathcal{S}$ is the cardinality of $\mathcal{C}$ (i.e., $l$), because only the $\{e_{mn} | \forall mn \in \mathcal{C} \}$ are the free variables. Second, as $\mathrm{rank}(\mat{A}')=N - 1 + k + K - l$ and the number of variables is $N-1+k+K$, the dimension of the entire solution space is $(N-1+k+K) - (N - 1 + k + K - l) = l$. From the above two statements, the solution subspace and the entire solution space of $\mat{A}'\vec{x}'=\vec{b}'$ have the same dimension. From the uniqueness of the solution space of linear equation system, the $\mathcal{S}$ is the entire solution space of $\mat{A}'\vec{x}'=\vec{b}'$.

The $\mathcal{S}$'s condition $e_{xy}=0$, $\forall xy \in \mathcal{A} \setminus \mathcal{C}$ means that the original $\mat{A}\vec{x} = \vec{b}$ does not allow any actual p2p synchronization fault without a corresponding estimated p2p synchronization fault. In the absence of $K$ actual p2p synchronization faults, the unique solution $\mathcal{S}$ encompasses the unique correct solution of $\mat{A}\vec{x} = \vec{b}$. In the presence of $K$ actual p2p synchronization faults, there are two cases.
\begin{enumerate}
\item If $l \!=\! k \!=\! K$, $\mathcal{S}$ is the unique correct solution of $\mat{A}\vec{x} = \vec{b}$;
\item Otherwise, we must have $l < K$. As a result, the $\mat{A}\vec{x}=\vec{b}$ must have no solutions, because otherwise the fact that $\mathcal{S}$ allows $l$ non-zero actual p2p synchronization faults only conflicts with the fact that there are $K$ non-zero actual p2p synchronization faults.
\end{enumerate}
\end{proof}

Based on Lemma~\ref{lem:2}, the following proposition can be used to compute the lower bound of maximum resilience.

  \begin{proposition}
    A system is $K$-resilient if $\forall k \in [0, K]$, for any distribution of the $K$ actual p2p synchronization faults and any distribution of the $k$ estimated p2p synchronization faults,
    $\mathrm{rank}(\mat{A}') = N - 1 + k + K - l$, where $l \in [0, k]$ is the number of correctly positioned estimated p2p synchronization faults.
    \label{prop:sufficient}
  \end{proposition}

  \begin{proof}
    As $\mathrm{rank}(\mat{A}') = N - 1 + k + K - l$, from Lemma~\ref{lem:1}, the original $\mat{A}\vec{x}=\vec{b}$ either has no solutions or has a unique correct solution. We now analyze the cases considered in Definition~\ref{def:resilience-condition}:
    \begin{enumerate}
    \item When $k \in [0, K)$, since $k < K$, the solution of $\mat{A}\vec{x}=\vec{b}$ cannot be correct. Thus, the $\mat{A}\vec{x}=\vec{b}$ has no solutions.
    \item When $k=K$,
      \begin{enumerate}
      \item if the distribution of the $k$ estimated p2p synchronization faults is identical to the distribution of the actual synchronization faults, as the statement that $\mat{A}\vec{x}=\vec{b}$ has no solution must not be true (because the correct solution is a solution), $\mat{A}\vec{x}=\vec{b}$ must have a unique (and correct) solution.
      \item otherwise, since the distributions are different, the solution of $\mat{A}\vec{x}=\vec{b}$ cannot be correct. Thus, the $\mat{A}\vec{x}=\vec{b}$ has no solutions.
      \end{enumerate}
    \end{enumerate}
    In summary, $\mathrm{rank}(\mat{A}') = N - 1 + k + K - l$ ensures that the $K$-resilience condition is satisfied.
  \end{proof}

  \renewcommand{\algorithmicrequire}{\textbf{Given:}}
\renewcommand{\algorithmicensure}{\textbf{Output:}}
\renewcommand{\algorithmiccomment}[1]{// #1}

\begin{algorithm}[t]
  \caption{Compute a lower bound of maximum resilience}
\label{alg:lower-bound}
\small
\begin{algorithmic}[1]
\REQUIRE The number of nodes $N$
\ENSURE A lower bound of maximum resilience

\STATE $K=0$
\WHILE{$K \le (N-2)$}
\label{line:n-2}
\FOR{each distribution of the $K$ actual p2p synchronization faults among the $\frac{N(N-1)}{2}$ p2p synchronization sessions}
\label{line:foreach1}
\STATE $k=0$
\WHILE{$k\leq K$}
\FOR{each distribution of the $k$ estimated faults among the $\frac{N(N-1)}{2}$ p2p synchronization sessions}
\label{line:foreach2}
\STATE determine the value of $l$ (i.e., the number of correctly positioned estimated faults)
\IF{$\mathrm{rank}(\mat{A}')\neq N - 1 + k + K - l$}
\label{line:use-of-proposition}
\RETURN{$K-1$}
\ENDIF
\ENDFOR
\STATE $k=k+1$
\ENDWHILE
\ENDFOR
\label{line:endforeach}
\STATE $K = K + 1$
\ENDWHILE
\end{algorithmic}
\end{algorithm}

\begin{table*}
  \normalsize
  \begin{minipage}{\textwidth}
    \begin{equation}
      \mat{A} = 
      \begin{blockarray}{ccccccccc}
        &\hat{\delta}_{10}& \cdots &\hat{\delta}_{(N-1)0}&\hat{e}_{10}& \hat{e}_{12} & \cdots & \hat{e}_{1(N-1)} & \cdots \\
        \begin{block}{c(cccccccc)}
          n_0 \leftrightarrow n_1       &1                & \cdots                & \cdot                &1           & 0 & \cdots  & 0 & \cdots \\
          n_0 \leftrightarrow n_2       &0                & \cdots                & \cdot                &0           & 0  & \cdots & 0 & \cdots \\
          n_i \leftrightarrow n_j, \forall i,j \neq 1        & \vdots                 & \ddots                & \vdots                & \vdots           & \vdots & \ddots & \vdots & \ddots \\
          n_{N-2} \leftrightarrow n_{N-1}       &0                &\cdots                & \cdot                &0           & 0 & \cdots & 0 & \cdots \\
          n_1 \leftrightarrow n_2       & -1                & \cdots                & \cdot                &0           & 1 & \cdots & 0 & \cdots \\
          \vdots        & \vdots                 & \ddots                & \vdots                & \vdots           & \vdots  & \ddots & \vdots & \ddots \\
          n_1 \leftrightarrow n_{N-1}       & -1                & \cdots                & \cdot                &0           & 0 & \cdots & 1 & \cdots \\
        \end{block}
      \end{blockarray}.
      \label{eq:prop1}
    \end{equation}
    \hrule
  \end{minipage}
\end{table*}

  Based on Proposition~\ref{prop:sufficient}, Algorithm~\ref{alg:lower-bound} computes a lower bound of maximum resilience for any $N$-node system. Specifically, by starting with no synchronization faults (i.e., $K=0$), it increases $K$ by one in each step of the outer loop to check whether the $N$-node system is $K$-resilient. The condition of $K \le (N-2)$ in Line~\ref{line:n-2} is from Proposition~\ref{prop:n-2} that the system is not $K$-resilient if $K > (N-2)$. The loops from Line~\ref{line:foreach1} to Line~\ref{line:foreach2} will generate all possible combinations of the distributions of actual and estimated synchronization faults. In Line~\ref{line:use-of-proposition}, we check whether the sufficient condition in Proposition~\ref{prop:sufficient} is met. If not, the current value of $K$ has already exceeded the lower bound of maximum resilience. Thus, the algorithm returns $K-1$ as the lower bound.

Table~\ref{tab:lower-bound} shows the results computed by Algorithm~\ref{alg:lower-bound} for $N$ up to 12. We can see that the lower bound of maximum resilience is a non-decreasing function of $N$, which is consistent with intuition. We also compute the lower bound of tolerance as $f_l(N) / \frac{N(N-1)}{2}$, i.e., the percentage of the faulty p2p synchronization sessions to ensure correct network clock synchronization. The last row of Table~\ref{tab:lower-bound} shows the lower bound of tolerance.



\subsection{Upper Bounds of Maximum Resilience}
\label{subsec:upper-bounds}

\begin{proposition}
  $f_u(N) = N - 2$ is an upper bound of maximum resilience, i.e., any $N$-node system is not $K$-resilient when $K > (N-2)$.
  \label{prop:n-2}
\end{proposition}

\begin{proof}
  We prove by an counterexample where all the $N-1$ p2p synchronization sessions involving the node $n_1$ are faulty. The remaining $K - (N - 1)$ faulty p2p synchronization sessions may occur between any other node pairs. Consider that Algorithm~\ref{alg:error-correction} is testing a distribution of the $K$ p2p synchronization faults that is identical to the actual distribution. Since the true clock offsets and the true p2p synchronization faults must form a valid solution to the equation system, we have $\mathrm{rank}(\mat{A} | \vec{b}) = \mathrm{rank}(\mat{A})$.

The matrix $\mat{A}$ of the vectorized equation system is given by Eq.~(\ref{eq:prop1}). We add labels to help understanding each column's corresponding unknown to be solved and each row's corresponding p2p synchronization session. In the first column of $\mat{A}$ that corresponds to the clock offset estimate $\hat{\delta}_{10}$, the first element and the last $N-2$ elements that correspond to all p2p synchronization sessions involving $n_1$ are non-zeros; all other elements are zero. This column is a linear combination of the columns corresponding to $\hat{e}_{10}, \hat{e}_{12}, \ldots, \hat{e}_{1(N-1)}$. Thus, $\mat{A}$ is not full column ranked. Therefore, the equation system $\mat{A}\vec{x} = \vec{b}$ have an infinite number of solutions, which violates the resilience condition.
\end{proof}

\section{Conclusion and Future Work}
\label{sec:conclude}

This paper studies how many p2p synchronization faults that an $N$-node system can tolerate in achieving network clock synchronization. Table~\ref{tab:lower-bound} gives the lower bound of maximum resilience under certain settings of $N$. We also prove that $N-2$ is an upper bound of maximum resilience.

It is interesting to study the following issues not addressed in this paper:
\begin{enumerate}
\item The tight bound of maximum resilience is still an open issue. However, even if the upper bound given by Proposition~\ref{prop:n-2} is tight, the tolerance $(N-2) / \frac{N(N-1)}{2}$ still decreases with $N$ when $N \ge 4$. It suggests that increasing the number of nodes is not beneficial in terms of fault tolerance. In future work, we will study how to reduce the number of p2p synchronization sessions and examine whether doing so can improve the fault tolerance.
\item Algorithm~\ref{alg:error-correction} and our analysis do not exploit the property that each fault is a multiple of $T$. If this discrete property is used, intuitively, the fault tolerance can be improved.
\end{enumerate}


\bibliographystyle{IEEEtran}
\bibliography{ref}

\begin{thebibliography}{10}
\providecommand{\url}[1]{#1}
\csname url@samestyle\endcsname
\providecommand{\newblock}{\relax}
\providecommand{\bibinfo}[2]{#2}
\providecommand{\BIBentrySTDinterwordspacing}{\spaceskip=0pt\relax}
\providecommand{\BIBentryALTinterwordstretchfactor}{4}
\providecommand{\BIBentryALTinterwordspacing}{\spaceskip=\fontdimen2\font plus
\BIBentryALTinterwordstretchfactor\fontdimen3\font minus
  \fontdimen4\font\relax}
\providecommand{\BIBforeignlanguage}[2]{{%
\expandafter\ifx\csname l@#1\endcsname\relax
\typeout{** WARNING: IEEEtran.bst: No hyphenation pattern has been}%
\typeout{** loaded for the language `#1'. Using the pattern for}%
\typeout{** the default language instead.}%
\else
\language=\csname l@#1\endcsname
\fi
#2}}
\providecommand{\BIBdecl}{\relax}
\BIBdecl

\bibitem{mills1991internet}
D.~L. Mills, ``Internet time synchronization: the network time protocol,''
  \emph{IEEE Trans. Commun.}, vol.~39, no.~10, pp. 1482--1493, 1991.

\bibitem{4579760}
``Ieee standard for a precision clock synchronization protocol for networked
  measurement and control systems,'' \emph{IEEE Std 1588-2008 (Revision of IEEE
  Std 1588-2002)}, pp. 1--300, July 2008.

\bibitem{elson2002fine}
J.~Elson, L.~Girod, and D.~Estrin, ``Fine-grained network time synchronization
  using reference broadcasts,'' \emph{ACM SIGOPS Operating Systems Review},
  vol.~36, no.~SI, pp. 147--163, 2002.

\bibitem{ganeriwal2003timing}
S.~Ganeriwal, R.~Kumar, and M.~B. Srivastava, ``Timing-sync protocol for sensor
  networks,'' in \emph{SenSys}.\hskip 1em plus 0.5em minus 0.4em\relax ACM,
  2003, pp. 138--149.

\bibitem{maroti2004flooding}
M.~Mar{\'o}ti, B.~Kusy, G.~Simon, and {\'A}.~L{\'e}deczi, ``The flooding time
  synchronization protocol,'' in \emph{SenSys}.\hskip 1em plus 0.5em minus
  0.4em\relax ACM, 2004, pp. 39--49.

\bibitem{rfc7384}
T.~Mizrahi, ``Security requirements of time protocols in packet switched
  networks,'' 2014, \url{https://tools.ietf.org/html/rfc7384}.

\bibitem{mizrahi2012game}
------, ``A game theoretic analysis of delay attacks against time
  synchronization protocols,'' in \emph{International Symposium on Precision
  Clock Synchronization for Measurement Control and Communication}, 2012.

\bibitem{ullmann2009delay}
M.~Ullmann and M.~V{\"o}geler, ``Delay attacks -- implication on ntp and ptp
  time synchronization,'' in \emph{International Symposium on Precision Clock
  Synchronization for Measurement, Control and Communication}, 2009.

\bibitem{viswanathan2016exploiting}
S.~Viswanathan, R.~Tan, and D.~K. Yau, ``Exploiting power grid for accurate and
  secure clock synchronization in industrial iot,'' in \emph{RTSS}.\hskip 1em
  plus 0.5em minus 0.4em\relax IEEE, 2016, pp. 146--156.

\bibitem{rabadi2017taming}
D.~Rabadi, R.~Tan, D.~K. Yau, and S.~Viswanathan, ``Taming asymmetric network
  delays for clock synchronization using power grid voltage,'' in
  \emph{AsiaCCS}.\hskip 1em plus 0.5em minus 0.4em\relax ACM, 2017, pp.
  874--886.

\bibitem{dongare2017pulsar}
A.~Dongare, P.~Lazik, N.~Rajagopal, and A.~Rowe, ``Pulsar: A wireless
  propagation-aware clock synchronization platform,'' in \emph{RTAS}, 2017.

\bibitem{chen2011ultra}
Y.~Chen, Q.~Wang, M.~Chang, and A.~Terzis, ``Ultra-low power time
  synchronization using passive radio receivers,'' in \emph{IPSN}, 2011.

\bibitem{nighswander2012gps}
T.~Nighswander, B.~Ledvina, J.~Diamond, R.~Brumley, and D.~Brumley, ``Gps
  software attacks,'' in \emph{CCS}.\hskip 1em plus 0.5em minus 0.4em\relax
  ACM, 2012, pp. 450--461.

\bibitem{yan2017}
Z.~Yan, Y.~Li, R.~Tan, and J.~Huang, ``Application-layer clock synchronization
  for wearables using skin electric potentials induced by powerline
  radiation,'' in \emph{SenSys}, 2017.

\bibitem{lukac2009recovering}
M.~Lukac, P.~Davis, R.~Clayton, and D.~Estrin, ``Recovering temporal integrity
  with data driven time synchronization,'' in \emph{SenSys}, 2009.

\bibitem{gupchup2009sundial}
J.~Gupchup, R.~Mus{\u{a}}loiu-e, A.~Szalay, and A.~Terzis, ``Sundial: Using
  sunlight to reconstruct global timestamps,'' in \emph{EWSN}, 2009, pp.
  183--198.

\bibitem{li2017natural}
Y.~Li, R.~Tan, and D.~K. Yau, ``Natural timestamping using powerline
  electromagnetic radiation.'' in \emph{IPSN}, 2017, pp. 55--66.

\bibitem{rowe2009low}
A.~Rowe, V.~Gupta, and R.~R. Rajkumar, ``Low-power clock synchronization using
  electromagnetic energy radiating from ac power lines,'' in
  \emph{SenSys}.\hskip 1em plus 0.5em minus 0.4em\relax ACM, 2009, pp.
  211--224.

\bibitem{li2012flight}
Z.~Li, W.~Chen, C.~Li, M.~Li, X.-Y. Li, and Y.~Liu, ``Flight: Clock calibration
  using fluorescent lighting,'' in \emph{MobiCom}.\hskip 1em plus 0.5em minus
  0.4em\relax ACM, 2012.

\bibitem{hao2011}
T.~Hao, R.~Zhou, G.~Xing, and M.~Mutka, ``Wizsync: Exploiting wi-fi
  infrastructure for clock synchronization in wireless sensor networks,'' in
  \emph{RTSS}, 2011, pp. 149--158.

\bibitem{li2011exploiting}
L.~Li, G.~Xing, L.~Sun, W.~Huangfu, R.~Zhou, and H.~Zhu, ``Exploiting {FM}
  radio data system for adaptive clock calibration in sensor networks,'' in
  \emph{MobiSys}.\hskip 1em plus 0.5em minus 0.4em\relax ACM, 2011, pp.
  169--182.

\bibitem{dolev1984possibility}
D.~Dolev, J.~Halpern, and H.~R. Strong, ``On the possibility and impossibility
  of achieving clock synchronization,'' in \emph{PODC}, 1984.

\bibitem{Lamport1985}
L.~Lamport and P.~M. Melliar-Smith, ``Synchronizing clocks in the presence of
  faults,'' \emph{JACM}, vol.~32, no.~1, pp. 52--78, 1985.

\bibitem{Shafarevich2012}
I.~R. Shafarevich and A.~Remizov, \emph{Linear algebra and geometry}.\hskip 1em
  plus 0.5em minus 0.4em\relax Springer Science \& Business Media, 2012.

\end{thebibliography}







\appendix

\subsection{Sensing-based P2P Synchronization Faults}
\label{appendix:reasons}

\subsubsection{Time fingerprinting approaches}

The studies \cite{viswanathan2016exploiting,li2017natural} show that the cycle length (i.e., $T$) of the power grid voltage \cite{viswanathan2016exploiting} and the associated powerline EMR \cite{li2017natural} has transient minute fluctuations over time in the order of 500 parts per million. The fluctuations at the same time in a geographic area served by the same power grid (e.g., a city) are nearly identical. Thus, a vector of successive cycle lengths is a time fingerprint. By matching a time fingerprint captured by $A$ against $B$'s historical time fingerprints that are timestamped respectively according to their clocks, the offset between $A$'s and $B$'s clocks can be estimated with a potential error of $nT$. If the time fingerprint length is sufficiently long, empirical zero error probability has been achieved \cite{viswanathan2016exploiting,li2017natural}. However, the possibility of errors cannot be precluded.


\subsubsection{Dirac-assisted NTP approaches}

As illustrated in Fig.~\ref{fig:principle}, the Dirac-assisted NTP transmits a {\em request} packet and a {\em reply} packet and records the transmission and reception timestamps $t_1$, $t_2$, $t_3$, and $t_4$ according to $A$'s and $B$'s clocks. It also computes the elapsed clock times for $t_1$, $t_2$, $t_3$, and $t_4$ from their respective last impulses (LIs) in the Dirac combs. These elapsed clock times (i.e., phases) are denoted by $\phi_1$, $\phi_2$, $\phi_3$, and $\phi_4$, as illustrated in Fig.~\ref{fig:principle}. The round-trip time (RTT) is $\mathrm{RTT} = (t_4 - t_1) - (t_3 - t_2)$. Define the {\em rounded phase differences} $\theta_q$ and $\theta_p$ (which correspond to the request and reply packets, respectively) as
\begin{equation*}
  \small
\theta_{q} \!=\! \left\{
\begin{array}{l}
\phi_2 \!-\! \phi_1, \text{ if } \phi_2 \!-\! \phi_1 \!\ge\! 0;\\
\phi_2 \!-\! \phi_1 \!+\! T, \text{ otherwise}.
\end{array}
\right.
\theta_{p} \!=\! \left\{
\begin{array}{l}
\phi_4 \!-\! \phi_3, \text{ if } \phi_4 \!-\! \phi_3 \!\ge\! 0;\\
\phi_4 \!-\! \phi_3 \!+\! T, \text{ otherwise}.
\end{array}
\right.
\end{equation*}
As analyzed in \cite{rabadi2017taming,yan2017}, we have $\mathrm{RTT} = \theta_q + \theta_p + (i + j) \cdot T$, $i,j \in \mathbb{Z}_{\ge 0}$, where the non-negative integers $i$ and $j$ are the numbers of elapsed periods of the external signals during the transmissions of the request and reply packets, respectively. For instance, in Fig.~\ref{fig:principle}, $i=2$ and $j=1$. Once the unknown $i$ or $j$ can be determined, the offset between $A$'s and $B$'s clocks can be estimated. However, solving $i$ and $j$ from $\mathrm{RTT} = \theta_q + \theta_p + (i + j) \cdot T$ is an integer-domain underdetermined problem that generally has multiple solutions. Arbitrarily choosing one of the candidate solutions will result in a clock offset estimation error of $nT$. The studies \cite{rabadi2017taming} and \cite{yan2017} proposed approaches to effectively reduce the number of candidate solution. However, it is challenging to ensure no ambiguity.

\end{document}